\documentclass[conference]{IEEEtran}
\IEEEoverridecommandlockouts

\usepackage{textcomp}
\usepackage{cite}
\usepackage{xurl}
\usepackage[hidelinks,breaklinks=true]{hyperref}
\usepackage{amsmath,amssymb,amsfonts}
\usepackage{algorithm}
\usepackage{algpseudocode}
\renewcommand{\algorithmicrequire}{\textbf{Input:}}

\usepackage{graphicx}
\usepackage{booktabs}
\usepackage{subfig}
\usepackage{makecell}
\usepackage{multirow}
\usepackage[framemethod=TikZ]{mdframed}
\usepackage{diagbox}
\usepackage{pifont}
\usepackage{tcolorbox}
\usepackage{xcolor}
\usepackage{enumitem}
\usepackage{tabularx}
\usepackage{array}
\usepackage{doi}
\usepackage{threeparttable}
\usepackage[inkscapelatex=false]{svg}

\newtheorem{definition}{Definition}
\newtheorem{theorem}{Theorem}
\newtheorem{proposition}{Proposition}
\newtheorem{corollary}{Corollary}



\newcommand{\Hist}{\mathcal{H}}
\newcommand{\Keys}{\mathcal{K}}
\newcommand{\Vals}{\mathcal{V}}
\newcommand{\Window}{\mathcal{W}}
\newcommand{\Intervals}{\mathcal{I}}

\newcommand{\Canon}{\mathrm{canon}}

\newcommand{\From}{\mathrm{From}}
\newcommand{\To}{\mathrm{To}}

\newcommand{\Source}{\mathcal{S}}

\def\BibTeX{{\rm B\kern-.05em{\sc i\kern-.025em b}\kern-.08em
    T\kern-.1667em\lower.7ex\hbox{E}\kern-.125emX}}

\begin{document}
\title{Vivace: Exact Temporal OLAP over Interval Histories via Independent Serverless Execution}

\author{
\IEEEauthorblockN{Woohyeok Park}
\IEEEauthorblockA{\textit{Department of Data Science} \\
\textit{Hanyang University}\\
Seoul, Republic of Korea \\
woohyeok@hanyang.ac.kr}
\and
\IEEEauthorblockN{Taeyoon Kim}
\IEEEauthorblockA{\textit{Department of Data Science} \\
\textit{Hanyang University}\\
Seoul, Republic of Korea \\
tykim7@hanyang.ac.kr}
\and
\IEEEauthorblockN{Hyunjoon Kim}
\IEEEauthorblockA{\textit{Department of Data Science} \\
\textit{Hanyang University}\\
Seoul, Republic of Korea \\
hyunjoonkim@hanyang.ac.kr}
\and
\IEEEauthorblockN{Kyungyong Lee\textsuperscript{\textdagger}}
\IEEEauthorblockA{\textit{Department of Data Science} \\
\textit{Hanyang University}\\
Seoul, Republic of Korea \\
kyungyong@hanyang.ac.kr}
}

\maketitle

\begin{abstract}
Temporal online analytical processing (OLAP) analyzes past states of data whose values change over time. Such histories are naturally stored as interval histories, in which each row records the period during which a value remained valid. Because temporal analyses typically arrive in infrequent, intermittent bursts, serverless execution that launches functions only at query time offers a cost advantage over always-on clusters. Splitting a computation that a single process performs as a whole across independent serverless functions, however, breaks correctness in two ways. A function may not receive the rows that determine the state of its time range, and naively summing partial results yields incorrect answers for duration-weighted and cumulative-threshold queries. Existing SQL engines and serverless analytics do not address both problems together.

This paper presents Vivace, a serverless system for exact temporal OLAP over interval histories. Vivace resolves the two problems in separate stages. Before any query arrives, a pre-query layout step partitions the interval history, replicating boundary-crossing intervals so each function computes its range completely from a single file. At query time, a merge step combines partial results under operator-specific rules. Associative aggregates merge intermediate values, and ranking re-orders candidates within each time range. We prove that this partitioned execution matches single-process computation up to canonical form. Evaluated on AWS Lambda with real-world datasets, Vivace reduces latency and monetary cost by up to 82\% and 84\%, respectively, against an equivalent SQL baseline that queries the history directly, demonstrating robust generality and efficiency.
\end{abstract}


\section{Introduction}
\label{sec:introduction}
\emph{Temporal online analytical processing (OLAP)} analyzes past states over the history of data that changes over time, supporting tasks such as post-incident investigation, quarterly market comparison, and regulatory audits over prices, inventory, and equipment configurations. The accumulated changes are commonly stored as an \emph{interval history}~\cite{scd}, which records each period of an unchanged value as a single \([\From,\To)\) row and thus preserves a long history at far lower storage cost than per-time-point snapshots. Because temporal analyses typically exhibit intermittent access patterns~\cite{serverless-db-lambada, skyrise, pixel-turbo}, an always-on cluster incurs high idle costs, whereas the \emph{serverless} execution, which launches only at query time, offers a compelling cost advantage for sporadic workloads~\cite{serverless-db-lambada,scalable-query-starling,cackle,pywren,spes-faas-trade-off}.

The benefit of serverless execution is maximized when functions execute independently without sharing state. However, splitting a monolithic temporal computation across independent functions breaks correctness. Naively partitioning interval data by time range loses information for rows crossing boundaries, preventing functions from reconstructing the true state. Furthermore, simply summing partial results from multiple parallel function executions breaks duration-weighted averages and cumulative counts. Sharing state between functions can compensate, but the added latency and data-movement costs offset the serverless advantage~\cite{storage-functions}. Therefore, exact independent execution requires two conditions. \emph{Input completeness} requires each function to receive every row affecting its time range, and \emph{Operator-correct merge} requires partial results to be combined using semantics-aware aggregation rules.

However, existing approaches satisfy at most one of the two. Object-storage SQL engines such as Athena~\cite{athena} and BigLake~\cite{biglake} can read interval histories, but users must hand-write the overlap detection and boundary clipping logic in every query. Serverless analytics engines such as Lambada~\cite{serverless-db-lambada} and Starling~\cite{scalable-query-starling} provide cost-efficient function execution over object storage, but they target only general relational data and offer neither specialized temporal representation nor merge techniques that interval overlap, duration weighting, and per-time-point counting require. Timeline Index~\cite{timeline-index} and ParTime~\cite{partime} compute temporal analytics efficiently in-memory, which is costly for large, rarely-accessed histories. 

This paper proposes \emph{Vivace}, a serverless temporal OLAP system that satisfies both requirements while keeping every function execution independent. Vivace resolves each requirement in a separate stage. During the pre-query phase, a \emph{layout} step divides the interval history at fixed time boundaries, splitting every interval that spans a boundary, so that each partition file holds the complete state of its own time range. Each function therefore reads a single file and obtains every row its range needs (\emph{input completeness}). At query time, each function evaluates the operator over its own range and emits a compact partial result. A reducer then merges these partial results under operator-specific rules and applies the final computation only once (\emph{operator-correct merge}). Neither stage shares state between functions, and we prove as a theorem that the merged result equals the single-process result up to canonical form, which is a property we call \emph{exactness}.

To the best of our knowledge, Vivace is the first system to characterize the conditions under which temporal OLAP over interval histories decomposes exactly across independent serverless functions. While the underlying techniques are standard, no prior work establishes when the combination preserves single-process semantics. We implemented Vivace on public-cloud FaaS and evaluated it on real datasets with diverse interval characteristics~\cite{spotlake-iiswc, caiso_oasis, wikimedia_mediawiki_history}. Every operator reproduced the single-process ground truth exactly. For a SQL baseline of identical semantics, Vivace eliminated the input-preparation cost that grows with history depth, reducing latency and cost.

\smallskip
\noindent\textbf{Contributions.} The contributions of this paper are as follows.

\begin{itemize}[leftmargin=*,itemsep=3pt,topsep=3pt]

\item \textbf{Problem formulation.}
We formalize \emph{input completeness} and \emph{operator-correct merge}, the conditions under which independent functions compute temporal OLAP over an interval history, each from its own input alone (\S\ref{sec:model}, \S\ref{sec:segment}).
\item \textbf{System design.}
Vivace pre-partitions the interval history by time so that each function computes its range from a single file, and merges partial results under operator-matched rules. It supports predicate windows, duration-weighted aggregation, counts over time, and cross-entity comparison (\S\ref{sec:overview}--\S\ref{sec:operators}).
\item \textbf{Exactness proof.}
For every supported analysis, we prove that partitioned parallel execution equals single-process computation up to canonical form (\S\ref{sec:exactness-theorem}).
\item \textbf{Empirical evaluation.}
On three datasets with different change rates, Vivace reduces latency by up to 82\% and cost by up to 84\% against a SQL baseline, and the same operator families answer every query of an independent SCD Type 2 benchmark without new primitives (\S\ref{sec:eval}).
\end{itemize}

\section{Pitfalls of Serverless Temporal OLAP}
\label{sec:background}
Operational data such as airline ticket prices, e-commerce inventory, and equipment configurations change continuously, and their histories accumulate. Since snapshotting every measurement time grows the row count as entities~$\times$~time points, the common practice is to record a new row only when a value changes. Lakehouse engines such as Iceberg and Delta Lake expose change histories through their table formats, storing them in object storage over long periods~\cite{iceberg-time-travel,delta-cdf}.

\begin{figure}[t]
\centering
\subfloat[Serverless cost model.]{
  \includegraphics[width=0.45\columnwidth]{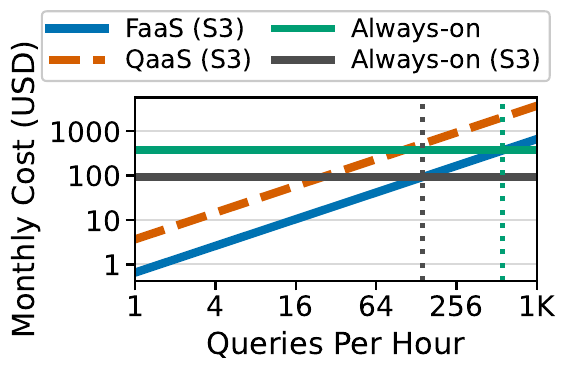}
  \label{fig:motivation-serverless}
}
\hfill
\subfloat[Coupled vs.\ Independent execution.]{
  \includegraphics[width=0.47\columnwidth]{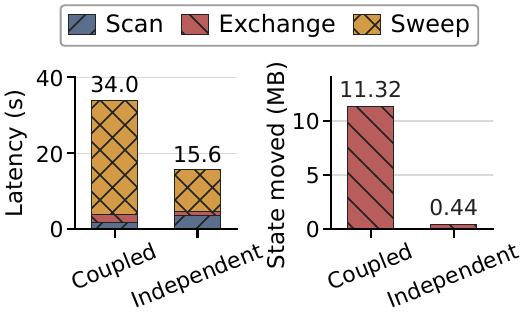}
  \label{fig:motivation-coupled-independent}
}
\caption{Motivating characteristics of serverless temporal OLAP.}
\label{fig:motivation}
\end{figure}

\subsection{Temporal OLAP over Slowly Changing Dimensions}
\label{sec:scd}
We call a history stored around its change points an \emph{\textbf{interval history}}, where each value change closes the previous interval and opens a new one, so each row holds one entity's value together with the interval during which it persisted. For example, if a product's price became \$10 on March~18 and held until May~3, the entire period is one row. This corresponds to \emph{slowly changing dimension (SCD)} Type~2 in data warehouse modeling~\cite{scd}. We denote a row as \((K, V, [\From, \To))\), where \(K\) is the primary key and \(V\) its value during the interval, and call a set of such rows the \emph{valid-from/to source relation}.

\emph{Temporal OLAP} extends traditional OLAP, which aggregates measures over dimensions for decision support~\cite{multidim-db,olap-models-survey}, to interval-valued queries over historical time ranges. For example, answering ``during which intervals in the second quarter was the price at most \$10?'' requires examining the change points within the requested range.
Temporal database research has established \emph{sequenced semantics}, which folds per-time-point query results into \([\From, \To)\) rows, and \emph{coalescing}, which merges adjacent rows with identical payloads into maximal intervals, as the standard for result comparison~\cite{jensen-unification,dignos-sequenced}. These analyses involve overlap detection, boundary-aligned clipping, duration-weighted aggregation, and per-time-point counting, whose cost grows as the history accumulates.

\subsection{Serverless Analytics}
\label{sec:serverless-bg}
For histories that accumulate over long periods but are read only occasionally, an always-on cluster continuously pays for data loading, index and cache maintenance, and idle resources. The Function-as-a-Service (FaaS) approach instead launches functions only at query time, reads only the required data in parallel, and bills only actual execution~\cite{serverless-db-lambada, scalable-query-starling, spes-faas-trade-off}. As Fig.~\ref{fig:motivation-serverless} shows, FaaS and Query-as-a-Service (QaaS) costs scale with query count while always-on costs do not, so the gap widens sharply when analysis is infrequent.\footnote{Based on a query reading a 1\,GiB compressed interval history. FaaS cost is measured; QaaS and always-on paths use us-west-2 public prices (June 2026), with the instance sized by the measured memory footprint.}

\subsection{Execution Models for Interval-History Analytics}
\label{sec:exec-models}
Existing systems compute exact interval-history analytics through the built-in distributed state exchange of managed services~\cite{athena,biglake}, through explicit intermediate state exchanged via shared storage in FaaS engines~\cite{serverless-db-lambada,scalable-query-starling}, or through user-written overlap/clipping SQL. All of these incur additional data movement or coordination at query time, and none supports the \emph{independent} execution, in which each worker reads the single file of its time range, computes its result without exchanging temporal state among workers, and ships only compact partials to the reducer. The difficulty of independent execution surfaces in queries that compare multiple entities on the same time axis, such as finding the lowest-priced entity at each time point in a 10K-entity cohort.

Fig.~\ref{fig:motivation-coupled-independent} decomposes the latency of this query into \emph{Scan} (locating intervals visible in the window), \emph{Exchange} (moving intermediate state through shared storage), and \emph{Sweep} (comparing intervals within each time bucket), along with the intermediate state crossing the mapper--reducer boundary (12-month 10K cohort, identical input and Lambda configuration). The \emph{Coupled} plan follows existing FaaS engines~\cite{serverless-db-lambada,scalable-query-starling}, where workers re-write the visible intervals into time-range buckets in shared storage, and the reducer reads each bucket to rank entities per time point. Intervals crossing bucket boundaries are replicated, inflating the reducer input to more than twice the source intervals. This state movement and replica reprocessing dominate total latency. The \emph{Independent} plan instead materializes the same time division as a pre-query layout. Because each worker's file already holds its intervals clipped to its range, the worker ranks entities within its own input and ships only a compact output to a reducer. For the same query, intermediate state movement thus drops by 25$\times$ and total latency falls below half. Even for cross-entity queries, closing temporal state inside the mappers reduces both the shared-storage exchange and the reducer-side reprocessing.

\subsection{Requirements for Partitioned Temporal OLAP}
\label{sec:requirements}
Consider the fan-out structure~\cite{pywren} of Fig.~\ref{fig:manual-fails}. The query window is divided into time ranges, one mapper handles each range, and a reducer collects the partial results. Product A costs \$10 during \([\text{Mar}\,18, \text{May}\,3)\) and \$20 during \([\text{May}\,3, \text{Jun}\,1)\). With each row stored in the file of the month containing its start time, the \$10 row lands in March, the \$20 row in May, and April holds nothing. Computing ``the average price of A during \([\text{Apr}\,1, \text{Jun}\,1)\)'' exactly imposes two requirements.

\begin{figure}[t]
\centering
\includegraphics[width=\columnwidth]{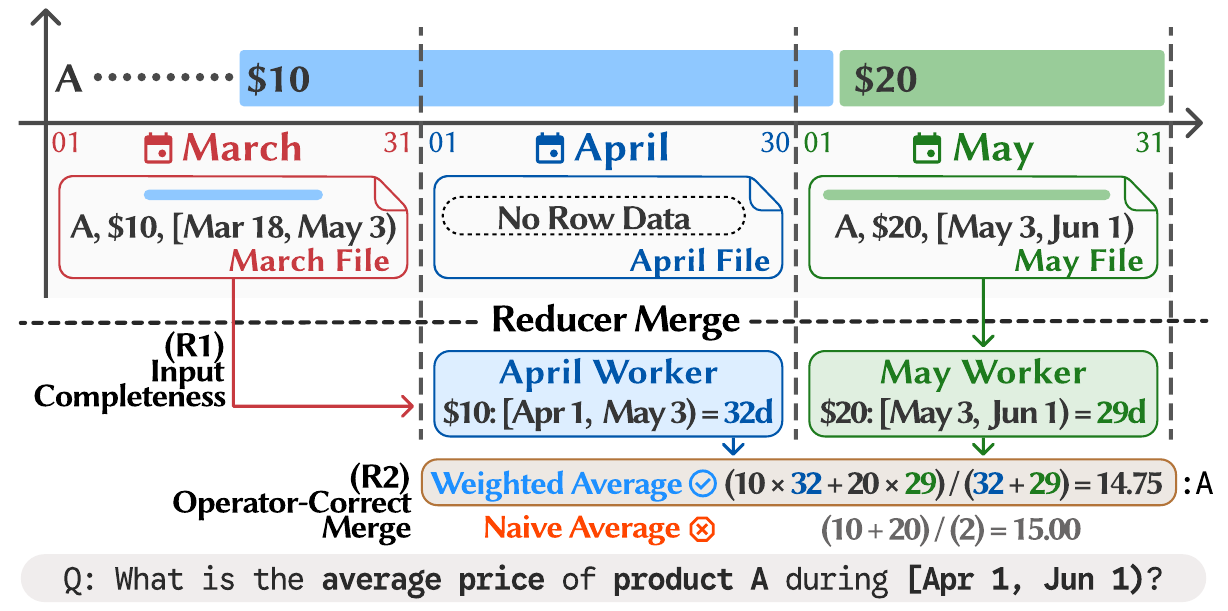}
\caption{Naive serverless fan-out breaks at two points.}
\label{fig:manual-fails}
\end{figure}

\begin{itemize}[leftmargin=*,itemsep=3pt,topsep=3pt]
\item \emph{Input completeness} (R1):
Each mapper must receive every row that determines the state of its range. Here the \([\text{Mar}\,18, \text{May}\,3) \to \$10\) row sits in the March file yet determines the April price, so the April mapper misses it unless it also reads March.

\item \emph{Operator-correct merge} (R2):
Partial results must be combined by rules matching the operator's semantics. The correct average over \([\text{Apr}\,1, \text{Jun}\,1)\) weights by duration, \((10\times32+20\times29)/(32+29)=14.75\), whereas averaging the per-mapper averages gives \((10+20)/2=15.00\). The same problem arises for count thresholds and ranking.
\end{itemize}

Finally, a fan-out result may fragment intervals at mapper boundaries, so comparison with the single-process result requires a \emph{canonical form} that merges adjacent identical ranges.

\section{Vivace Overview}
\label{sec:overview}
Vivace adopts independent execution where each mapper reads only the single file of its time range and ships results to the reducer, with no communication between mappers. This independence minimizes query-time data movement and preserves serverless cost advantages. To satisfy R1 and R2 while preserving this independence, Vivace divides query processing into three steps (Fig.~\ref{fig:Vivace-system-overview}). Before any query, the \emph{layout} step restructures the source relation so that every mapper's input is complete, satisfying R1. At query time, each \emph{mapper} evaluates its range and emits a partial result; the \emph{reducer} merges these partials under operator-specific rules, satisfying R2.

\begin{figure}[t]
\centering
\includegraphics[width=\columnwidth]{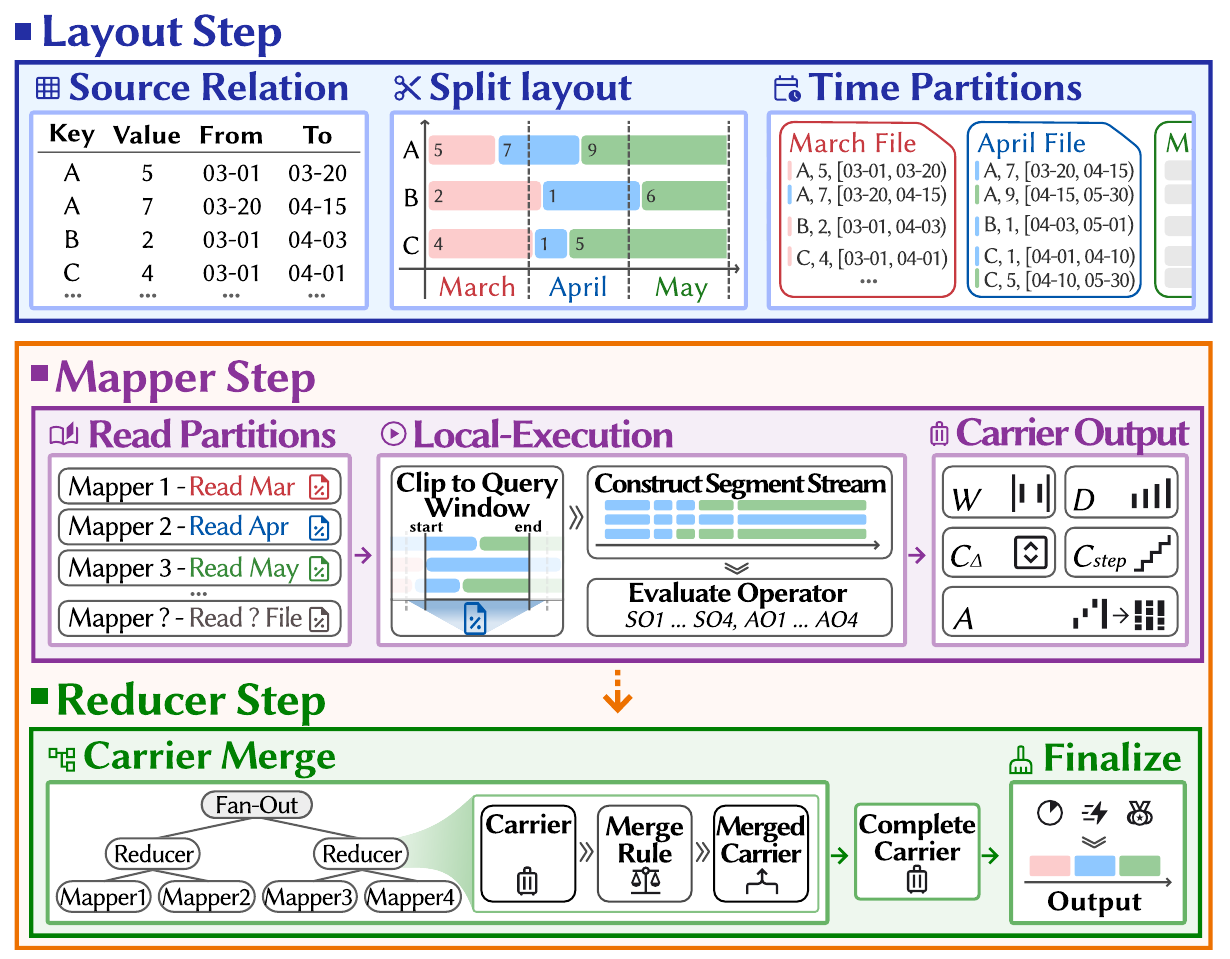}
\caption{Vivace system overview.}
\label{fig:Vivace-system-overview}
\end{figure}

\paragraph{Layout}
Vivace transforms the source relation into a \emph{time-partitioned layout} once, before any query arrives. An interval crossing partition boundaries is split and stored in every file it spans. For example, in a monthly layout, the \([\text{Mar}\,18, \text{May}\,3)\to \$10\) row enters the March, April, and May files. Each range between two layout boundaries is a \emph{time chunk}, and by default the planner assigns one chunk to one mapper, whose input thus holds the complete state of its range.

\paragraph{Mapper}
A mapper splits its range at every point where a value changes, producing a time-ordered \emph{segment stream}. No value changes within a segment, so the per-segment \emph{evaluator} performs predicate evaluation, duration accumulation, count-delta derivation, or ranking comparison. The results reach the reducer as a \emph{Carrier}, an intermediate format fixed per operator family.

\paragraph{Reducer}
Because naive summation violates R2 for some operators, each Carrier holds a combinable partial result and a finalize rule applied once after merging. For duration-weighted averages, the weighted sums and durations are added across mappers, and the division is performed only at finalize. Chunk-fragmented results are then coalesced into canonical form for comparison with the single-process result.

\section{Data Model and Operators}
\label{sec:model}
Vivace handles the history of the values entities hold over time. Arranging the rows $(K,V,[\From,\To))$ of one entity in time order yields a timeline of when the entity held which value. We write this entire timeline as the \emph{interval history} $\Hist$, the set of possible entities as $\Keys$, and the values as $\Vals$. The value entity $K$ holds at time $\tau$ is written $\mathsf{state}_{\Hist}(K,\tau)$, and since rows of the same entity never overlap in time, this value is unique at each time point. Over this data we define the layout, mapper, and reducer as a formal model.
 
Storing this history directly as rows yields the \emph{valid-from/to source relation} $\Source$. We assume each row holds the full state of its entity rather than a subset of columns, and call such rows \emph{full-state rows}. A source that records only the changed columns is first reconstructed into full-state rows. A single row thus reveals every attribute of the entity at that time, so multiple predicates and aggregates can be evaluated over it.
 
Given a query window $\Window=[t_s,t_e)$, we write $\Intervals(\Hist,\Window)$ for the rows that overlap the window, clipped to its boundaries. After clipping, rows of the same entity still do not overlap, and each row's $V$ remains constant within its interval. All operators in this paper operate over this $\Intervals(\Hist,\Window)$.

 
\begin{table*}[t]
\centering
\caption{Operator families and mapper outputs.}
\label{tab:operator-family}
\begin{tabularx}{\textwidth}{@{}l l l l X@{}}
\toprule
\textbf{Family} & \textbf{Segment input} & \textbf{Per-segment computation} & \textbf{Mapper output} & \textbf{Final result} \\
\midrule
\textbf{SO1} predicate window & same-base & $\varphi(V)$ & interval window $W$ & per-key intervals where $\varphi$ holds \\
\textbf{SO2} duration aggregate & same-base & $x(V) \cdot (T-F)$ & duration partial $D$ & duration-weighted aggregate by group \\
\textbf{SO3} boolean composition & same-base & $\varphi_a \wedge \varphi_b$, $\vee$, $\wedge \neg$ & interval window $W$ & intersection, union, difference intervals \\
\textbf{SO4} count timeline & same-base & endpoint delta $\pm 1$ & $C_\Delta \to C_{step}$ & group count timeline; threshold windows \\
\textbf{AO1} pairwise compare & aligned pair & $V_{left} \prec V_{right}$ & aligned $A$ (boolean) & comparison timeline between two entities \\
\textbf{AO2} count-better-than & aligned cohort & $|\{j : V_j \prec V_{ref}\}|$ & aligned $A$ (count) & count timeline over better candidates \\
\textbf{AO3} winner / top-$k$ & aligned cohort & ranking per segment & aligned $A$ (ranked) & winner or top-$k$ timeline \\
\bottomrule
\end{tabularx}
\end{table*}
 
\subsection{Segment Stream}
\label{sec:segment}
The rows a mapper receives from its partition file cannot be used directly for computation over time ranges. Multiple entities mix at the same time points, and the query window may cut through the middle of a time chunk. The mapper therefore cuts its range wherever a value changes and treats the span between two adjacent cuts as one segment.
 
\begin{definition}[Segment stream]
\label{def:segment-stream}
Let $P = [a,b)$ be the time range a mapper is responsible for, $\Window = [t_s,t_e)$ the query window, and $(K,V,[F,T))$ a mapper input row. A row that overlaps both $P$ and $\Window$ is clipped as
\begin{equation}
(K,V,[F',T')),\,
F'=\max\{F,a,t_s\},\,
T'=\min\{T,b,t_e\}.
\end{equation}
The stream of clipped rows with $F' < T'$, sorted in ascending order of start time, is called the \emph{segment stream} $S(P,\Window)$.
\end{definition}
 
Segments of the same key do not overlap in time and $V$ is constant within one segment, so every operator family takes this stream as a common input and processes it in time order.
 
\subsubsection{Closed Segment Partition}
 
When multiple mappers share the window, their results combine to equal the whole only if each mapper fully reproduces the state of its assigned range and keys from its own input. Each mapper piece, the unit dividing the window, is thus closed within its own input, and we call this condition the \emph{closed segment partition}.
 
Suppose the query window $\Window$ is divided into non-overlapping mapper pieces, where each piece $(j,s)$ handles time chunk $P_j$ and key subset $K_s$, and the pieces cover every time and key of $\Window$.
 
\begin{definition}[Closed segment partition]
\label{def:closed-partition}
A mapper piece division $\{(j,s)\}$ forms a \emph{closed segment partition} if the local segment stream $S_{j,s}$ of each piece equals the restriction of the global segment stream $\Intervals(\Hist,\Window)$ to that piece. If $K_s = \Keys$, the partition is a time-only partition. Formally, for every $K \in K_s$ and $\tau \in P_j$,
\begin{equation}
\mathsf{state}{S{j,s}}(K,\tau)=\mathsf{state}_{\Hist}(K,\tau).
\end{equation}
\end{definition}

\subsubsection{Building the Segment Stream}
 
A segment stream is built in one of two ways. The first projects a single interval relation directly into a segment stream. Since the same row carries the full state of all attributes, multiple predicates, aggregates, and boolean compositions are evaluated over the same segments.
The second aligns the entities of a \emph{cohort} $C = \{K_1,\ldots,K_m\}$, the entity set a query designates for comparison, onto common ranges. Every time point where any entity's value changes becomes a segment boundary, so the entire cohort holds constant values within each segment and the evaluator can compare them together.

\subsection{Time-Partitioned Layout}
\label{sec:layout}
 
The time-partitioned layout is the mappers' input, the source relation cut before queries arrive. It is built over predetermined time boundaries $b_0 < b_1 < \cdots < b_M$, and each range between adjacent boundaries $[b_\ell,b_{\ell+1})$ is a \emph{time chunk}. The chunk width $\Delta_{chunk} = b_{\ell+1} - b_\ell$ is fixed at build time.
 
\begin{definition}[Time-partitioned layout]
\label{def:layout}
A layout over time boundaries $b_0 < b_1 < \cdots < b_M$ places a partition file $\Pi_\ell$ for each time chunk $[b_\ell,b_{\ell+1})$, where
\begin{equation}
\begin{split}
\Pi_\ell = \bigl\{(K,V,[\max(F,b_\ell),\min(T,b_{\ell+1}))) \mid \\
(K,V,[F,T))\in\Source,\ F<b_{\ell+1},\ T>b_\ell\bigr\}.
\end{split}
\end{equation}
Every row in $\Pi_\ell$ has its time interval contained in $[b_\ell,b_{\ell+1})$.
\end{definition}
 
Under this definition, a boundary-crossing source row is clipped and enters every partition file it spans, so the closed segment partition of \S\ref{sec:segment} holds automatically once a mapper receives only its own partition file. In SQL/file baselines that scan the source relation directly, the mapper itself locates the rows overlapping its chunk to meet the same condition. We call this layout \emph{interval-clipped partitioning} (ICP). Each partition file directly holds the $(K,V,[\From,\To))$ rows clipped to its time chunk, so a mapper reads it as the input of its range without further transformation.

\subsubsection{Time-Chunk Granularity and Layout Size}
\label{sec:chunk-size}
 
The time chunk size $\Delta_{chunk}$ jointly determines the storage volume and the mapper input size. ICP clips boundary-crossing intervals into every chunk they span, so as chunks shrink, one source interval is replicated across more chunks and the stored rows grow. We measure this inflation by the row amplification $A_{row}$ and byte amplification $A_{byte}$ relative to the source, and select the chunk width to keep both within budget.
 
\paragraph{Lifetime and chunk width}
A source interval of length $\ell$ crosses the boundaries of chunks of width $w$ ($=\Delta_{chunk}$) about $\ell/w$ times and is therefore clipped into about $1+\ell/w$ rows. Averaging over all intervals gives
\begin{equation}
A_{row}(w) \approx 1 + \frac{\bar\ell}{w}
\end{equation}
where $\bar\ell$ is the mean interval lifetime. Amplification thus tracks the ratio of mean lifetime to chunk width. With short lifetimes, amplification stays small even under frequent changes. With long lifetimes, the same interval is clipped into many chunks even under rare changes, inflating storage.
 
\paragraph{Storage-safe time width}
Let $w_{safe}$ be the smallest chunk width keeping amplification within the budgets $\epsilon_R,\epsilon_B$.
\begin{equation}
w_{\mathrm{safe}}=\min\{w: A_{\mathrm{row}}(w)\le1+\epsilon_R
\wedge A_{\mathrm{byte}}(w)\le1+\epsilon_B\}.
\end{equation}
Chunks finer than $w_{safe}$ exceed the budget, so $w_{safe}$ is the lower bound on the chunk width, and the mapper input size is fitted by dividing primary keys rather than shrinking chunks.

\subsection{Operator Families}
\label{sec:operators}
 
Implementing each complex temporal query separately would require re-fitting the mapper input and the merge rules for every query type. Vivace instead defines a small set of basic operators and expresses more complex analyses as their compositions. As long as the basic operators guarantee fan-out execution and exact merging, queries composed from them inherit the same guarantee. Each operator reads a segment stream and produces an interval result or a timeline, and Table~\ref{tab:operator-family} summarizes the segment input, mapper output, and final result of each family.
 
Predicate windows (SO1 and SO3), duration-weighted summaries (SO2), and count timelines (SO4) are \emph{same-base operators}, whose results depend only on the $V$ of a single row (\S\ref{sec:same-base-ops}). Aligned cohort timelines (AO1--AO3) are \emph{aligned-cohort operators}, which compare multiple entities of a cohort within the same segment (\S\ref{sec:aligned-ops}). The input is restricted to a single interval relation, or its join with time-invariant dimensional attributes.
 
Analyses beyond this scope are not supported, because their results do not close over the mapper-local segment stream alone. An $N$-way temporal join requires inputs from multiple relations~\cite{overlap-interval-join}, a dynamic comparison cohort changes during execution, a sequential pattern needs ordering state across segment boundaries, and a sliding-window query needs overlap across partitions~\cite{arraystore,scalable-temporal-agg-icde00} or prefix/suffix state~\cite{prefix-sum-cikm2024}.

\subsubsection{Same-Base Operators}
\label{sec:same-base-ops}
 
\paragraph{Predicate window (SO1)}
For a predicate $\varphi:\Vals \to \{0,1\}$, SO1 is defined as
\begin{equation}
\mathrm{SO}_1(\Intervals,\varphi) = \{(K,V,[F,T)) \in \Intervals \mid \varphi(V) = 1\}.
\end{equation}
The result is the full-state intervals where the condition holds.
 
\paragraph{Duration-weighted aggregation (SO2)}
Given a function $G(V)$ that determines the aggregation group from a row's value $V$ and an expression $x(V)$ that extracts the quantity to aggregate, the duration-weighted average of group $g$ is
\begin{equation}
\mathrm{TWA}_g = \frac{\sum_{i:G(V_i)=g} x(V_i)\Delta_i}{\sum_{i:G(V_i)=g}\Delta_i},\quad \Delta_i = T_i - F_i.
\end{equation}
The mapper outputs the duration partial before the final ratio,
\begin{equation}
D_g = \bigl(\textstyle\sum x_i\Delta_i,\ \sum\Delta_i,\ n,\ \min x_i,\ \max x_i\bigr),
\end{equation}
which carries the weighted sum, the total duration, the segment count, and the extremes of $x$ for group $g$. The ratio is computed once after the $D_g$ of all mappers are collected.
 
\paragraph{Boolean composition (SO3)}
Predicates $\varphi_a, \varphi_b$ are evaluated over the same full-state row, reducing computation to a boolean expression on one row without interval intersection.
\begin{align}
\mathrm{SO}_{3I} &= \{r \in \Intervals \mid \varphi_a(r.V) \wedge \varphi_b(r.V)\}, \\
\mathrm{SO}_{3U} &= \{r \in \Intervals \mid \varphi_a(r.V) \vee \varphi_b(r.V)\}, \\
\mathrm{SO}_{3D} &= \{r \in \Intervals \mid \varphi_a(r.V) \wedge \neg \varphi_b(r.V)\}.
\end{align}
 
\paragraph{Count timeline and threshold (SO4)}
The number of entities in group $g$ for which the predicate holds at time $t$ is

\begin{equation}
\label{eq:count-def}
\resizebox{0.99\columnwidth}{!}{$
\displaystyle
N_g(t)=\bigl|\{(K,V,[F,T))\in\Intervals \mid G(V)=g,\ \varphi(V),\ F\le t<T\}\bigr|
$}
\end{equation}

The count timeline is the set of maximal intervals $(g,[F,T),N_g(F))$ over which $N_g(t)$ is constant. The threshold result keeps only the intervals with $N_g(t) > c^*$. The mapper emits per-group $+1/-1$ deltas $\Delta_g$ as $C_\Delta$ at each endpoint of the intervals where the predicate holds, and the reducer builds the full count step timeline $C_{step}$ by cumulative summation in time order.

\subsubsection{Aligned-Cohort Operators}
\label{sec:aligned-ops}
Comparing multiple entities at the same time points requires their segment boundaries to coincide.
To this end, all interval endpoints of the cohort $C$, the entity set the query designates, become common boundaries, and clipping each entity's history at these boundaries yields the \emph{aligned segment stream}.
Within one aligned segment, every entity of $C$ holds a constant value, so per-segment comparison is possible.
Aligned-cohort operators apply their evaluators over this stream.
 
\begin{definition}[Cohort-closed input]
\label{def:cohort-closed}
The input of a mapper piece is \emph{cohort-closed} with respect to cohort $C$ if it contains the interval rows of every entity needed to reproduce the aligned segment stream of $C$ within the piece's time chunk.
\end{definition}

\paragraph{Pairwise compare (AO1)}
Under comparator $\prec$, the comparison result of two entities $h_a, h_b$ is
\begin{equation}
\mathrm{Cmp}(h_a, h_b, x, \prec) = \{[F,T) \mid x_a(t) \prec x_b(t)\ \forall t \in [F,T)\}.
\end{equation}
 
\paragraph{Count-better-than (AO2)}
For a reference entity $h_r$ and candidates $\{h_j\}$, the number of candidates better than the reference at time $t$ is
\begin{equation}
\mathrm{CB}(t) = |\{j \mid x_j(t) \prec x_r(t)\}|.
\end{equation}
 
\paragraph{Winner and top-$k$ (AO3)}
In each aligned segment, the active values within the cohort are ordered by a deterministic comparator and ranks $1, \ldots, k$ are output. The winner timeline is the case $k = 1$.
 
\subsubsection{Composition}
Operators compose within and across families. For example, a predicate window or a boolean result feeds a count timeline, and a threshold is then applied. Thresholds and ranking are final computations applied once after all mapper results gather at the reducer. A composition that restricts one window to the key set identified by another predicate is handled by the \emph{key-domain finalizer}, which applies the restriction once after both windows are complete.

\section{Serverless Execution and Correctness}
\label{sec:exact}
In serverless fan-out execution, each mapper evaluates the segment stream of one range of the query window and emits partial results as Carriers, which the reducer combines into the final result. Matching the single-process result requires input completeness (R1) and operator-correct merge (R2). We call this equality \emph{exactness} and prove that it holds.
 
Vivace executes a query in the three stages of Algorithm~\ref{alg:execution} (Fig.~\ref{fig:vivace-implementation-overview}) and is available as open source.\footnote{\url{https://github.com/ddps-lab/vivace}} For each query $Q$, the plan fixes, by operator family, the segment evaluator $m_Q$, the merge rule $M_Q$, the finalize map $\mathrm{fin}_Q$, and the canonicalizer $\Canon_Q$. Carriers are written to object storage and completions to the commit registry. The reducer starts only after every mapper task has committed. Commits are idempotent, so a re-executed mapper's result is reflected only once.
 
\begin{figure}[t]
\centering
\includegraphics[width=\columnwidth]{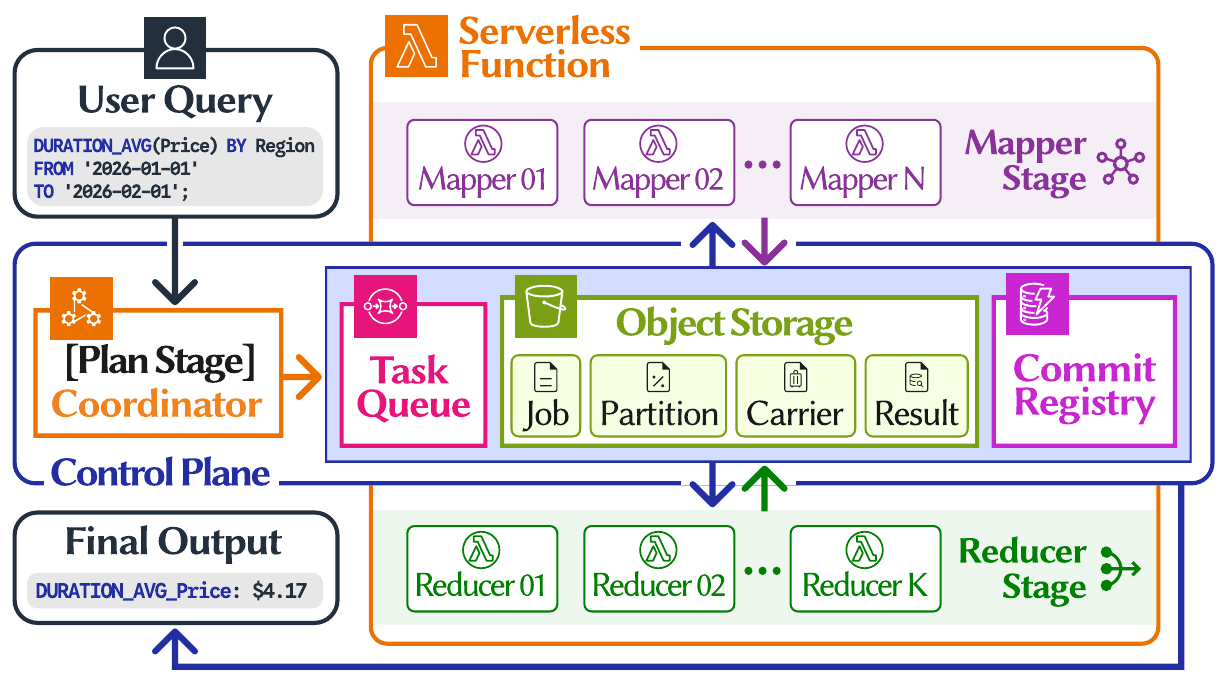}
\caption{Vivace implementation overview.}
\label{fig:vivace-implementation-overview}
\end{figure}

\begin{algorithm}[t]
\caption{Fan-out execution of query $Q$ over window $\Window$}
\label{alg:execution}
\small
\begin{algorithmic}[1]
\Require query $Q$ over window $\Window$;
\Statex \hphantom{\algorithmicrequire\ }ICP partition files $\Pi_j$ per time chunk $P_j$
\Ensure canonicalized result of $Q$
\For{each time chunk $P_j$ overlapping $\Window$}\Comment{\textbf{Plan stage}}
    \State divide $P_j$ into mapper pieces $(j,s)$, $s=1,\ldots,q_j$
    \State register one mapper task per piece in the task queue
\EndFor
\ForAll{tasks $(j,s)$ \textbf{in parallel}}\Comment{\textbf{Mapper stage}}
    \State read mapper input $X_{j,s}$ from partition file $\Pi_j$; clip to $P_j \cap \Window$
    \State build the segment stream $S_{j,s}$ 
    \State write Carrier $C_Q^{j,s} \gets m_Q(S_{j,s})$; commit task $(j,s)$
\EndFor
\State wait until every task has committed\Comment{completion barrier}
\State $\bar{C}_Q \gets M_Q(\{C_Q^{j,s}\})$\Comment{\textbf{Reducer stage}: merge law}
\State \Return $\Canon_Q(\mathrm{fin}_Q(\bar{C}_Q))$\Comment{finalize once}
\end{algorithmic}
\end{algorithm}
 
\subsection{Query Planning and Mapper Execution}
\label{sec:plan-and-mapper}
 
\paragraph{Plan validation}
The planner converts a logical request into the tasks of Algorithm~\ref{alg:execution}, but first validates the execution conditions required by the merge laws of Table~\ref{tab:merge-laws}.
It rejects as non-executable: operators outside the supported scope, inputs without full-state payloads, plans missing the completion barrier, and ranked queries without a fixed candidate shard placement or a deterministic tie-breaking rule.

\paragraph{Bounded mapper pieces}
When one chunk is large relative to the mapper memory, shrinking the chunk width would inflate storage through boundary clipping, so the width is kept and the chunk's primary keys are divided by hash into \emph{shards}. For same-base queries, the mapper input $X_{j,s}$ of Algorithm~\ref{alg:execution} is the rows of its chunk whose key hashes to $s$.
 
Same-base operators are independent per key, so the closed segment partition is preserved as long as the shards are disjoint and cover all keys. Aligned-cohort queries divide the query's candidate set into disjoint shards instead of hash shards. Count-better-than replicates the reference history into each shard and sums the partial counts, and winner/top-$k$ has the reducer re-rank the per-shard top-$k$ to restore the global ranks. Without a separate physical layout, an aligned mapper reads only its own cohort keys from the same hash-sharded partition files, satisfying cohort-closed correctness without gathering the entire cohort into one mapper.
 
 
The shard count $q_j$ is the smallest value for which each shard's estimated compressed bytes fit within the calibrated ceiling $B_{fit}$, and the estimate comes from a per-PK byte histogram, reflecting key skew. $B_{fit}$ is the largest compressed piece size that a mapper completes without Out of Memory (OOM). Since the mapper decompresses and processes its piece in memory, $B_{fit}$ is far smaller than the memory capacity.

\paragraph{Carriers}
 
The mapper emits its partial results as \emph{Carriers}, intermediate results whose output schema, merge method, finalization point, and canonicalization method are fixed by the plan.
The reducer reads only Carriers, never mapper internal state, and applies the merge rule the plan records for each Carrier type.
Table~\ref{tab:carrier-types} summarizes the five Carrier types the runtime uses.

\begin{table*}[t]
\centering
\caption{Carrier types used by the runtime. Main users are the operator families defined in \S\ref{sec:operators}.}
\label{tab:carrier-types}
\setlength{\tabcolsep}{2pt}
\renewcommand{\arraystretch}{1.04}
 
\begin{tabularx}{\textwidth}{@{}
  >{\raggedright\arraybackslash}p{0.13\textwidth}
  >{\raggedright\arraybackslash}p{0.30\textwidth}
  >{\raggedright\arraybackslash}X
  >{\raggedright\arraybackslash}p{0.22\textwidth}
  @{}}
\toprule
\textbf{Carrier} & \textbf{Shape} & \textbf{Invariant} & \textbf{Main users} \\
\midrule
 
\(\boldsymbol{W}\) interval window
& \((K,F,T,[\mathit{payload}])\)
& key/time sorted; non-overlapping per key
& SO1, SO3; threshold windows \\
 
\(\boldsymbol{D}\) duration partial
& \((G,\textstyle\sum_i x_i\Delta_i,\sum_i\Delta_i,n,\min_i x_i,\max_i x_i)\)
& additive partial by group
& SO2 \\
 
\(\boldsymbol{C_\Delta}\) count delta
& \((G,\mathit{Time},\delta)\)
& deltas collapsed by group/time
& SO4 endpoint merge \\
 
\(\boldsymbol{C_{\mathit{step}}}\) step timeline
& \((G,F,T,\mathit{count})\)
& chunk-clipped or globally complete count function
& SO4 full timeline; threshold input \\
 
\(\boldsymbol{A}\) aligned timeline
& boolean/additive/ranked aligned intervals
& coalesced output timeline
& AO1--AO3 \\
 
\bottomrule
\end{tabularx}
\end{table*}
 
\subsection{Merge Laws}
\label{sec:merge-laws}
 
The reducer combines the mappers' Carriers to restore the original interval sets or the full count timeline.
A merge must not start from an incomplete Carrier set, and the completion barrier of Algorithm~\ref{alg:execution} guarantees this.

\begin{definition}[Complete carrier]
\label{def:complete-carrier}
When the set of mapper pieces $\{(j,s)\}$ of query $Q$ forms the closed segment partition of \S\ref{sec:segment} covering the entire query window, the carrier combined over it is called \emph{complete}.
\begin{equation}
\bar{C}_Q = M_Q\bigl(\{C_Q^{j,s}\}\bigr)
\end{equation}
\end{definition}
Table~\ref{tab:merge-laws} summarizes how each Carrier is combined and which conditions are required.
All rules presuppose a closed segment partition, a complete carrier set before finalization, a plan-fixed deterministic finalizer, and canonical output comparison, so the table lists only the rule-specific conditions.
A superscript $(j,s)$ marks one mapper piece's carrier, a bar marks the combined carrier, and $\Canon$ is the canonical coalescing of \S\ref{sec:exactness-theorem}.
 
\begin{table*}[t]
\centering
\caption{Merge rules and required conditions.}
\label{tab:merge-laws}
\renewcommand{\arraystretch}{1.2}
\begin{tabularx}{\textwidth}{@{}l l >{\raggedright\arraybackslash}X >{\raggedright\arraybackslash}X@{}}
\toprule
\textbf{Rule} & \textbf{Carrier} & \textbf{Merge step} & \textbf{Rule-specific condition} \\
\midrule
\textbf{L1} interval union & $W$, boolean $A$ & $\bar{W}=\Canon\bigl(\biguplus_{j,s} W^{j,s}\bigr)$; same-base boolean uses L1 after per-row evaluation & same-base predicates evaluate one full-state row; aligned boolean input is cohort-closed \\
\textbf{L2} additive duration partial & $D$ & $\bar{D}_g=\bigoplus_{j,s} D^{j,s}_g$ with $\oplus$ componentwise $(+,+,+,\min,\max)$; ratio once at finalize & complete duration partials before final ratio \\
\textbf{L3} count delta endpoint addition & $C_\Delta/C_{step}$ & $\bar{N}_g(t)=\sum_{\tau\le t}\sum_{j,s}\Delta^{j,s}_g(\tau)$; canonical coalesce & complete clipped endpoint coverage per group \\
\textbf{L4} threshold after full count & $C_{step}\rightarrow W$ & $\bar{W}=\Canon\bigl(\{(g,[F,T))\mid \bar{N}_g(t)>c^*\ \forall t\in[F,T)\}\bigr)$ & complete count timeline before threshold \\
\textbf{L5} aligned additive count merge & additive $A$ & $\overline{\mathrm{CB}}(t)=\sum_{j,s} \mathrm{CB}^{j,s}(t)$ by endpoint addition; canonical coalesce & fixed cohort; disjoint candidate shards; replicated reference history \\
\textbf{L6} ranked aligned merge & ranked $A$ & collect shard top-$k$ candidates; re-align; recompute global top-$k$ & complete disjoint candidate shards; total order including tie key \\
\textbf{L7} key-domain finalizer & tuple $(W, W_\psi)$ & merge both by L1; $\Canon\bigl(\bar{W} \ltimes_K \pi_K(\bar{W}_\psi)\bigr)$ & both carriers complete; same base key domain \\
\bottomrule
\end{tabularx}
\end{table*}
 
\paragraph{Interval union and duration partial (L1, L2)}
L1 and L2 are immediate over a closed segment partition. For L1, each mapper's keep/drop result covers disjoint time pieces, so union followed by coalescing restores the single-process result. Aligned boolean compare applies the same structure over the aligned Carrier $A$. For L2, the weighted sum, duration, and count are additive and min/max are associative, so combining mapper partials yields the global partial, and only the ratio is computed once after merging.
 
\paragraph{Count timeline and threshold barrier (L3, L4)}
The SO4 count timeline (\S\ref{sec:operators}) tracks the number of entities satisfying the predicate over time, and applies a threshold $c^*$ on top of it. Unlike L1 and L2, this merge is not immediate. For example, if two mappers counting one time point over key shards obtain 40 and 60, the sum is 100. Yet if the threshold $c^*=80$ is applied per mapper first, neither passes and the time point drops out of the result. The threshold must therefore be applied after all mapper counts are summed. The following proposition shows that this structure is exact.
 
\begin{proposition}[Count timeline merge and threshold barrier]
\label{prop:count-timeline}
SO4 is exact if the endpoint deltas of all mappers are summed by group and time, the full step timeline is built by cumulative summation, and the threshold is then applied.
\end{proposition}
 
\begin{IEEEproof}
Let the \emph{local count} $N^{j,s}_g(t)$ be the value mapper piece $(j,s)$ obtains by evaluating Eq.~\eqref{eq:count-def} over its local segment stream. By Definition~\ref{def:closed-partition}, the local stream of each piece equals the restriction of the global stream to its time chunk $P_j$ and key shard $s$. Since the shards are disjoint and their union covers the full key set, $\sum_s N^{j,s}_g(t)=N_g(t)$ for every $t \in P_j$, and the L3 endpoint addition performs this summation.
 
A mapper emits its local count as $C_\Delta$ endpoint deltas, or as chunk-clipped step segments $(g,[F,T),c)$ that the reducer re-expresses as the deltas $+c@F$ and $-c@T$. Summing the deltas of all mappers by group and time gives
\begin{equation}
\Delta_g(t)=\sum_{j,s} \Delta^{j,s}_g(t),
\end{equation}
and cumulative summation in time order restores the full count step timeline. The threshold is the final computation evaluated over the full count timeline.
\end{IEEEproof}
 
\paragraph{Aligned count merge (L5)}
L5 applies the count timeline merge to the candidate shards of an aligned cohort.
 
\begin{corollary}[Aligned count merge]
\label{cor:aligned-count}
Count-better-than (AO2) is expressed as the count payload of the aligned Carrier $A$, and merging by the same endpoint-addition scheme is exact under a fixed cohort, disjoint candidate shards, and a replicated reference history.
\end{corollary}
 
\begin{IEEEproof}
Count-better-than is an aligned count timeline without a group key. If the reference history accompanies every candidate shard and the shards are disjoint, the cumulative sum of the endpoint-addition merge equals the aligned count of a single process over the same cohort.
\end{IEEEproof}
 
\paragraph{Ranked aligned merge (L6)}
In aligned winner/top-$k$, each candidate belongs to exactly one shard. Within one aligned segment, an entity in the global top-$k$ is also in its own shard's top-$k$, because if its shard held $k$ or more better candidates, its global rank would also fall outside $k$. Re-sorting the per-shard top-$k$ Carriers with the same comparator and tie-breaking key therefore yields the global ranked result.

\paragraph{Key-domain finalizer (L7)}
The \emph{key-domain finalizer} restricts a same-base interval window result $W$ to the key set defined by a separate predicate $\psi$. Letting $E_\psi = \pi_K(\mathrm{SO}_1(\Intervals,\psi))$ be the set of keys satisfying $\psi$,
\begin{equation}
\mathrm{Gate}(W, \psi) = W \ltimes_K E_\psi = \{r \in W \mid r.K \in E_\psi\}.
\end{equation}
The mapper emits both $W$ and $W_\psi = \mathrm{SO}_1(\Intervals,\psi)$ as interval window carriers, and the reducer completes both by L1 and applies the semi-restriction $W \ltimes_K \pi_K(W_\psi)$ once. This finalizer is exact over complete carriers.
 
\begin{proposition}[Key-domain finalizer exactness]
\label{prop:key-domain-gate}
If the interval window carriers $W$ and $W_\psi$ are each completed by L1 over a closed segment partition, and $W^{mono}$ and $W_\psi^{mono}$ denote the two results a single process computes from the same input, then $W \ltimes_K \pi_K(W_\psi)$, applied after all pieces have arrived, equals $W^{mono} \ltimes_K \pi_K(W_\psi^{mono})$ in canonical form.
 
\end{proposition}
 
\begin{IEEEproof}
Let $M_W$ denote the L1 merge of interval window carriers. Then $M_W(\{W^{j,s}\})=W^{mono}$ and $M_W(\{W_\psi^{j,s}\})=W_\psi^{mono}$. Key projection preserves union, so over a closed partition $\bigcup_{j,s} \pi_K(W_\psi^{j,s})=\pi_K(W_\psi^{mono})$. Since $E_\psi$ is computed from the complete $W_\psi$, the membership evidence of a key may come from any time chunk or shard. The semi-restriction is therefore applied once after both carriers are complete. The semi-join keeps or drops all intervals of a key $K$ as a whole and introduces no new time boundaries, so it commutes with $\Canon$. We thus obtain $\Canon(W^{mono} \ltimes_K \pi_K(W_\psi^{mono}))$.
\end{IEEEproof}

\paragraph{Finalization}
After the merge, the reducer applies the query's final computation once.
 
\begin{definition}[Finalizer]
\label{def:finalizer}
Each query $Q$ has a finalize map $\mathrm{fin}_Q$, determined by its operator family, that sends a complete carrier or carrier tuple to the final result. $\mathrm{fin}_Q$ is a deterministic function of the complete carrier state alone, such as the per-group ratio for SO2, the threshold for SO4, the ranking selection for winner/top-$k$, and the key-domain semi-restriction.
\end{definition}
 
When $\mathrm{fin}_Q$ is the identity, the merged carrier is promoted to the result manifest immediately.

\subsection{Exactness Theorem}
\label{sec:exactness-theorem}
 
We now prove that the execution of Algorithm~\ref{alg:execution} produces the same result as a single process.
 
\paragraph{Canonical equality}
 
Two semantically identical results may divide their intervals differently. A fan-out execution may split one interval into two rows at a chunk boundary, whereas a single process emits the same range as one maximal interval. Row-by-row comparison therefore requires first arranging both results into the same canonical form.
 
\begin{definition}[Canonical output form]
\label{def:canonical}
For results carrying intervals, $\Canon_Q$ sorts the rows by their output identifier and time, then merges adjacent rows with $T_i=F_{i+1}$ that match on the comparison criterion into maximal intervals. The criterion differs by output family. A general interval window compares the key and $V$. A result that discards $V$, such as SO3-U, compares only the key. A count timeline compares the group and the count. An aligned timeline compares the aligned result, and a ranked timeline compares the rank together with the selected entity or value.
\end{definition}
 
With the comparison criteria fixed, the maximal interval decomposition is unique, so $\Canon_Q$ is idempotent and unaffected by input order or chunk fragmentation. Result equivalence is judged by
\begin{equation}
\Canon_Q(\mathrm{Result}_{\mathrm{fanout}})=\Canon_Q(\mathrm{Result}_{\mathrm{mono}}).
\end{equation}
For floating-point outputs such as duration-weighted aggregates, the group set, total duration, min, and max are compared exactly, and the TWA within a fixed tolerance.
 
The following theorem formalizes this equivalence.

\begin{theorem}[Closed-partition decomposition exactness]
\label{thm:exactness}
Suppose each row holds the full state of its key, each mapper receives every interval row that determines the active state within its chunk, and for aligned queries the mapper inputs satisfy the cohort-closed condition of \S\ref{sec:operators}. Suppose further that the operator family of query $Q$ satisfies the merge rules and execution conditions of Table~\ref{tab:merge-laws}. Then the fan-out result, obtained by merging the mappers' local evaluation results and applying the finalize map $\mathrm{fin}_Q$ and the canonicalizer $\Canon_Q$, equals the $\Canon_Q$ of the single-process result.
\end{theorem}
 
\begin{IEEEproof}
A mapper piece is a time chunk, optionally paired with a key shard. By Definition~\ref{def:closed-partition}, each piece's local segment stream $S_{j,s}$ equals the restriction of the global segment stream to that piece, so the clipped global stream is the disjoint union
\begin{equation}
S=\biguplus_{j,s}S_{j,s}.
\end{equation}
In the aligned families, cohort-closed inputs supply each mapper the needed entity histories, so the same decomposition holds over aligned segments.
 
Let each mapper output be $C_Q^{j,s}=m_Q(S_{j,s})$. Over a closed segment partition, $M_Q(\{C_Q^{j,s}\})$ is the complete carrier $\bar{C}_Q$ of Definition~\ref{def:complete-carrier}, and under the conditions of Table~\ref{tab:merge-laws} it equals the single process's intermediate result before finalization. Specifically, L1 covers SO1/SO3 and pairwise compare (AO1), L2 addresses SO2; Proposition 1 targets SO4; Corollary 1 handles count-better-than (AO2); L6 covers winner/top-k (AO3); and L7 is paired with Proposition~\ref{prop:key-domain-gate}. Since $\mathrm{fin}_Q$ depends only on the complete carrier or carrier tuple, applying it and then $\Canon_Q$ to both sides yields the same result.
\begin{equation}
\Canon_Q\bigl(\mathrm{fin}_Q(\bar{C}_Q)\bigr)=
\Canon_Q(Q(\Hist,\Window))
\end{equation}
\end{IEEEproof}

\section{Evaluation}
\label{sec:eval}
 
We empirically evaluate the Vivace implementation to address the following research questions.
 
\begin{description}[leftmargin=*, nosep]
\item[RQ1: Exactness.]
Does the fan-out execution result exactly match the result computed by a single process?
 
\item[RQ2: Layout efficiency.]
Does the ICP layout remove the mapper input preparation cost at only a small storage cost?
 
\item[RQ3: Layout portability.]
Can the layout be fitted on datasets with different interval characteristics without recalibrating mapper memory?
 
\item[RQ4: End-to-end efficiency.]
Does Vivace execute at lower monetary cost and latency than a SQL engine baseline?
 
\item[RQ5: Operator generality.]
Do the same operator families produce exact results on a dataset with a different schema?
\end{description}
 
\subsection{Experimental Setup}
\label{sec:eval-setup}

\paragraph{Environment}
All experiments, including the baseline engines and the S3 buckets for source relations and ICP layouts, run in AWS us-west-2. Vivace executes on Lambda with 5\,GB mapper functions, and comparisons within a figure use the same tier. Carriers are exchanged through S3, with SQS as the task queue and DynamoDB as the commit registry.
 
\paragraph{Dataset}
The primary dataset is the AWS spot market history of SpotLake~\cite{spotlake-iiswc}, a public dataset that has recorded the price, savings ratio, interruption frequency, and Spot Placement Score of spot instances for several years under the composite key of region $\times$ availability zone (AZ) $\times$ instance type. A new row at every value change forms one $[\From,\To)$ interval, making the dataset an interval history.
 
Analyses~\cite{multi-spotlake-www, multinode-spot-dataset, spotvista-arxiv, kubepacs} performed over this dataset, such as price-condition interval search, per-region duration-weighted averages, and instance type comparison, correspond directly to Vivace's operator families, so the operator results are the final analyses researchers want, without separate dimension joins. The dataset is also read intermittently, rarely in ordinary operation with bursts at failure investigations or research experiments, the access pattern where the serverless cost advantage is most pronounced. Its public availability facilitates reproduction. Evaluation centers on the 12-month window from March 2025 to March 2026 (Table~\ref{tab:eval-corpora}).

\paragraph{Additional datasets}
The evaluation adds two datasets whose interval characteristics differ from SpotLake.
CAISO LMP~\cite{caiso_oasis} is a public time series recording the hourly locational marginal price of the California ISO day-ahead market per pricing node (12-month evaluation window). Values refresh hourly, so interval lifetimes are short, while the key count and per-chunk data volume are similar to SpotLake.
Wikipedia revision history~\cite{wikimedia_mediawiki_history} records the revision size of English Wikipedia per page (3-month evaluation window). Most pages change little, so interval lifetimes are very long and the key count reaches tens of millions.
The three datasets differ widely in mean interval lifetime and per-chunk data volume (Table~\ref{tab:eval-corpora}), and we use this spread to verify that the layout sizing applies regardless of dataset characteristics.
MobilityDB-TPCDS, an SCD Type 2 benchmark, is used separately to assess operator generality.
 
\begin{table}[t]
\centering
\caption{Evaluation datasets and interval characteristics.}
\label{tab:eval-corpora}
\small
\begin{tabularx}{\columnwidth}{@{}lXrr@{}}
\toprule
\textbf{Dataset} & \textbf{Primary key} & \textbf{\#keys} & \textbf{Avg lifetime} \\
\midrule
SpotLake AWS  & region$\times$AZ$\times$type & 41{,}356 & 2.7\,h\\
CAISO LMP & pricing node & 67{,}892 & 1.3\,h \\
Wikipedia revision & page & 41.5\,M & $\sim$67\,d \\
\bottomrule
\end{tabularx}
\end{table}
 
\paragraph{Time-chunk and mapper configuration}
The source Parquet files are split by month, matching SpotLake's monthly ingestion cycle, and these monthly boundaries serve directly as time-chunk boundaries.
The monthly ICP partition files average 56.2\,MiB with a maximum of 61.6\,MiB, and in the default configuration each mapper reads one monthly chunk.

\paragraph{Baselines and execution paths}
The final end-to-end comparison focuses on two paths.
Athena SCD2 directly reads the natural valid-from/to source relation that users keep in the object store, and Vivace ICP reads the interval-clipped partition files, Vivace's intended layout.
Both paths' files are generated from the same source history and validated with the same query semantics and single-process oracle.
On SpotLake, Athena ICP and Vivace SCD2 are measured as an ablation to isolate the layout effect from the runtime effect.

Latency is recorded as the end-to-end time, the sum of the submit round trip and the terminal-status wait time.
Vivace cost sums function execution, object store I/O, and queue and commit registry requests, and Athena cost applies 5\,USD/TiB to scanned bytes with a 10\,MiB minimum, both computed from runtime billable units and official unit prices.
Log storage, intermediate objects, and post-execution artifact reads are excluded and recorded separately.

\begin{table*}[t]
\centering
\caption{Evaluation query suite and analysis examples.}
\label{tab:eval-query-suite}
\scriptsize
\setlength{\tabcolsep}{2pt} 
\begin{tabularx}{\textwidth}{@{}l l
>{\raggedright\arraybackslash}p{0.33\textwidth}
>{\raggedright\arraybackslash}p{0.3\textwidth}
>{\raggedright\arraybackslash}X@{}}
\toprule
\textbf{Query} & \textbf{Operator} & \textbf{Analysis question} & \textbf{Predicate or composition} & \textbf{Result semantics} \\
\midrule
Q1 & SO1 & When is spot capacity cheap? &
\texttt{SO1}(\texttt{SpotPrice} $<0.05$). &
Cheap-capacity intervals per entity. \\
Q2 & SO2 & What is each region's price exposure over the window? &
\texttt{SO2}(\texttt{SpotPrice} by \texttt{Region}). &
Region-level duration-weighted average price. \\
Q3 & SO3-I & When is capacity both cheap and stable? &
\texttt{SO3-I}(\texttt{SpotPrice} $<0.05$, \texttt{SPS} $\ge3$). &
Cheap-and-stable intervals. \\
Q4 & SO3-U & When is capacity cheap or stable? &
\texttt{SO3-U}(\texttt{SpotPrice} $<0.05$, \texttt{SPS} $\ge3$). &
Intervals satisfying at least one operational condition. \\
Q5 & SO3-D & When is capacity cheap but below the stability threshold? &
\texttt{SO3-D}(\texttt{SpotPrice} $<0.05$, \texttt{SPS} $\ge3$). &
Cheap intervals that still carry stability risk. \\
Q6 & SO4 & When does a region have enough cheap capacity? &
\texttt{SO4}(\texttt{SpotPrice} $<0.05$, \texttt{Region}, count $>100$). &
Regional market-depth intervals. \\
Q7 & AO1 & When is one AZ cheaper than another? &
\texttt{AO1}(\texttt{a1.2xlarge}, \texttt{aps1-az3} $\prec$ \texttt{aps1-az2}). &
Intervals where the left AZ is cheaper. \\
Q8 & AO2 & How many candidates beat a reference AZ? &
\texttt{AO2}(\texttt{a1.2xlarge}, min \texttt{SpotPrice} vs reference). &
Time-varying cheaper-candidate count. \\
Q9 & AO3 & Which candidate is cheapest at each time? &
\texttt{AO3}(\texttt{a1.2xlarge}, min \texttt{SpotPrice}, $k=1$). &
Cheapest-candidate timeline. \\
Q10 & AO3 & Which candidates are the top-$k$ cheapest over time? &
\texttt{AO3}(\texttt{a1.2xlarge}, min \texttt{SpotPrice}, $k=3$). &
Ranked cheapest-candidate timeline. \\
\bottomrule
\end{tabularx}
\end{table*}
 
\paragraph{Query suite}
Table~\ref{tab:eval-query-suite} summarizes the query suite, with the concrete analysis questions and predicates on SpotLake, the primary dataset.
The ten queries each use one core operator of \S\ref{sec:operators}, Q1--Q6 covering the same-base operators SO1--SO4 and Q7--Q10 the aligned-cohort operators AO1--AO3.
On CAISO and Wikipedia, the same ten queries are retained, with the predicates and group keys adapted to each domain, such as LMP/MCC or revision-size conditions.
 
For each target window, every system-query combination discards one warm-up run and takes the median of five repeats, reporting the mean of these medians over the ten queries. On SpotLake, the suite additionally runs over 1-, 3-, 6-, and 12-month windows to observe scaling with window length.
 
\subsection{Exactness Validation (RQ1)}
\label{sec:eval-exact}
 
To address RQ1, we compare the Vivace fan-out result with the single-process reference result in the canonical form of \S\ref{sec:exactness-theorem} for every operator family.
 
The canonical-form criterion differs by operator family. Interval window results (SO1, SO3-I/U/D) are compared at the row level after coalescing. SO2 aggregates compare the weighted sum, duration, min, and max exactly, with only the TWA ratio allowed a $10^{-9}$ tolerance. SO4 count timelines are canonicalized by group and count, with the threshold applied over the complete timeline, and aligned cohort operators are compared after canonicalizing the aligned payload.

The canonical forms matched for every query and window combination of Table~\ref{tab:eval-query-suite}.
The Athena baseline also produced the same results under the identical canonical-form criterion.

In the time-chunk sweep, the remote final Parquet results of five chunk factors ($\times1/4$, $\times1/2$, $\times1$, $\times2$, $\times4$) were read directly and compared with the monthly factor $\times1$ after per-carrier canonicalization.
All measured queries matched.
Notably, the raw output of SO1 varied from 12.12M to 12.26M rows across factors, yet after coalescing every factor agreed on the same 144,054 canonical intervals. These results confirm that exactness holds regardless of chunk division, and the subsequent performance comparisons rest on this premise.

\subsection{Layout Effectiveness (RQ2)}
\label{sec:eval-layout}
 
\begin{figure}[t]
\centering
\includegraphics[width=\columnwidth]{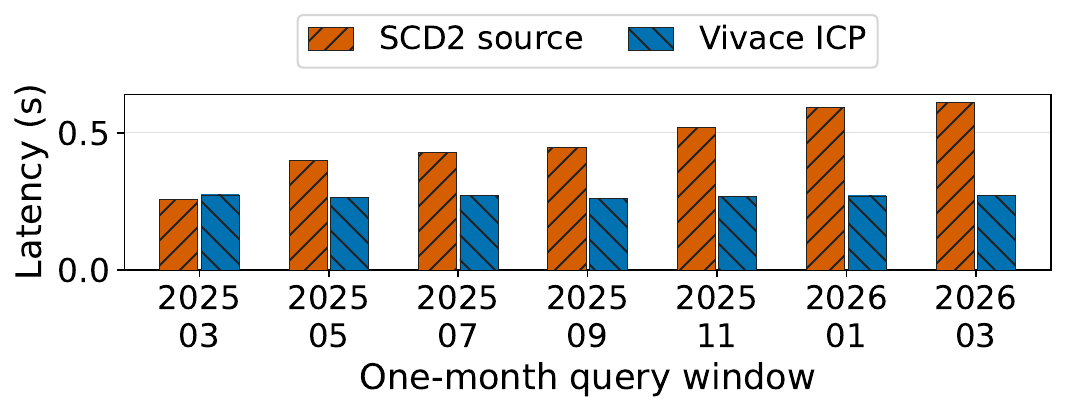}
\caption{SCD2 vs.\ ICP input preparation cost.}
\label{fig:eval-q1-input-form-latency}
\end{figure}
 
For RQ2, the input preparation cost and storage amplification of the ICP layout are measured against the SCD2 source.
 
\paragraph{Input preparation cost}
Fig.~\ref{fig:eval-q1-input-form-latency} compares the mapper latency of SCD2 and ICP while fixing a one-month query window and shifting the target month later. As the target month moves later, the number of prior monthly files SCD2 must read together grows.
The SO1 predicate window is the most basic operator and therefore exposes the mapper input preparation cost of the two input forms most directly.
 
SCD2 must locate the rows overlapping the query window in source files stored by \texttt{From} month. Both layouts show similar latency at first, where each reads a single file, but as prior files accumulate, SCD2 latency grows in proportion while ICP remains nearly flat. At the last window, SCD2 is about $2.25\times$ slower than ICP.
 
\paragraph{Storage amplification}
ICP clips boundary-crossing intervals into the files on both sides, so the row count and the total storage grow relative to the source relation. On the 12-month AWS dataset, the valid-from/to SCD2 source is 672.4\,MiB and the monthly ICP layout is 674.8\,MiB. ICP provides mapper-local input at $+0.3\%$ storage over the SCD2 source. The row count grows from boundary clipping, but the marginal storage difference is primarily attributed to Parquet's columnar compression. The $\From$ and $\To$ columns of clipped rows in ICP concentrate on the same time-chunk boundaries, raising the compression efficiency of dictionary and run-length encoding. As a result, the storage growth from added rows is offset by the compression gain in the timestamp columns.
 
These results confirm an asymmetry between the two layouts. ICP's storage increase is a one-time fixed cost, whereas SCD2's overlap lookup cost recurs at every query and grows in proportion to the history depth.
 
\subsection{Layout Portability (RQ3)}
\label{sec:eval-portability}
The layout predetermines the per-mapper fan-out, and fitting each mapper's input within its memory tier is Vivace's responsibility, different from query-time SCD2 scans, where this burden falls on the user. For RQ3, we test whether a single sizing model fits all three datasets, whose change rates and volumes differ widely, onto the same 5\,GB mapper tier.
 
\paragraph{Calibrated mapper piece ceiling}
The mapper piece ceiling $B_{fit}$ is calibrated only once, on SpotLake AWS. Materializing layouts at several byte targets and executing without shards, we take the largest target at which the six supported same-base queries finish exactly on a 5\,GB mapper without OOM or escalation to a larger memory tier. At larger targets the union/coalesce chunk exceeds the ceiling, but splitting the same chunk into two key shards brings it back within the ceiling. Aligned-cohort queries use cohort-closed placement, so $B_{fit}$ is defined on the same-base class. The resulting $B_{fit}=64$\,MiB is applied to every dataset.
 
\paragraph{Sizing outcomes}
Sizing sets the chunk width at or above $w_{safe}$ to bound storage and then divides primary keys by hash until each mapper piece fits within $B_{fit}$.
For SpotLake and CAISO, one chunk already fits within $B_{fit}$, so one mapper processes it without shards ($q=1$). For Wikipedia, revisions span the entire measurement window, so $w_{safe}$ equals the whole window. The chunk width therefore cannot be reduced, and a single chunk reaches 989\,MiB, about 15$\times$ $B_{fit}$. The model keeps this chunk as a single interval file clustered into primary-key buckets, and query-time logical sharding assigns each mapper only the bucket ranges its query needs, keeping the piece it reads within the memory tier. The six same-base queries ran over this query-time bucket sharding, and the four aligned-cohort queries over separate cohort-closed placement. All ten produced results identical to the single process on 5\,GB mappers without OOM.

\paragraph{Optimality}
This sizing satisfies the mapper memory constraint while avoiding fragmentation, and minimizing latency or cost is left as a separate problem. These results validate that the once-calibrated model generalizes, fitting all three datasets to the same mapper tier without recalibration.
 
\subsection{End-to-End Cost and Latency (RQ4)}
\label{sec:eval-cost-latency}
 
\begin{figure}[t]
\centering
\includegraphics[width=0.52\columnwidth]{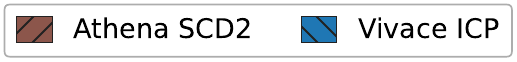}\\[-10pt]
\subfloat[Latency]{%
  \includegraphics[width=0.49\columnwidth]{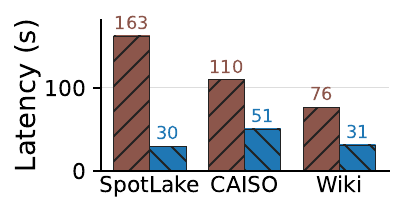}%
  \label{fig:eval-final-v2-latency}}\hfill
\subfloat[Cost]{%
  \includegraphics[width=0.49\columnwidth]{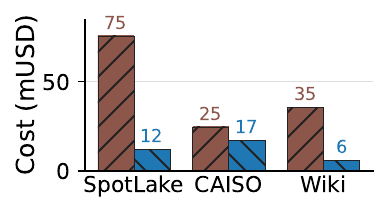}%
  \label{fig:eval-final-v2-cost}}
\caption{End-to-end latency and cost over the ten-query suite.}
\label{fig:eval-final-3dataset-bars}
\end{figure}

To address RQ4, Vivace is compared with SQL engine baselines in end-to-end cost and latency. Fig.~\ref{fig:eval-final-3dataset-bars} presents the per-dataset comparison, where Athena reads the natural SCD2 source and Vivace reads the ICP layout.
Vivace lowers the aggregate latency and cost of the ten queries on all three datasets. The gap is largest on SpotLake, where the history is deepest, at $5.44\times$ in latency and $6.23\times$ in cost. CAISO and Wikipedia follow at $2.17\times$ and $2.44\times$ in latency, with 30\% and 83\% lower cost.
These gains incur almost no storage overhead, with a Vivace-to-Athena storage ratio of 0.91--1.03.

\begin{figure}[t]
\centering
\includegraphics[width=\columnwidth]{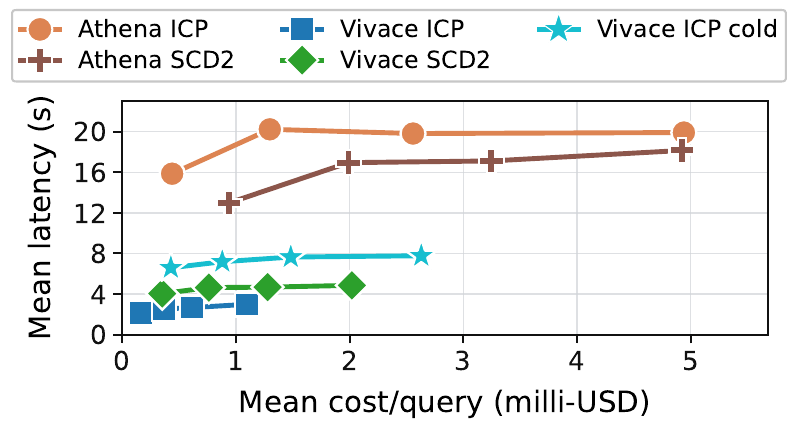}
\caption{SpotLake cost and latency across query windows}
\label{fig:eval-cost-latency}
\end{figure}

\paragraph{SpotLake deep dive}
On SpotLake we examine scaling by widening the query suite over 1-, 3-, 6-, and 12-month windows. Fig.~\ref{fig:eval-cost-latency} presents the per-query average cost in milli-USD (mUSD) on the x-axis and the latency on the y-axis for five execution paths, with each path's leftmost point at the 1-month window. Paths closer to the origin execute at lower cost and latency.
BigQuery, measured under the same SQL contract, is excluded because its logical-scan-bytes billing makes the per-query cost far higher than Athena's. Notably, Athena ICP exceeds Athena SCD2 in both latency and cost. Athena does not assign mappers per file, so ICP only adds scan volume without its partition-local benefit.
Vivace ICP's latency grows only from 2.10\,s to 2.98\,s as the window widens from 1 to 12 months. The mapper count grows with the window, but parallel execution keeps the latency growth sublinear. The curve stays closer to the origin than Athena SCD2 throughout. In contrast, Vivace SCD2's overlap lookup of prior monthly files grows with window length, and at 12 months it is 63\% slower and about $1.8\times$ more expensive than ICP.
Under forced cold starts, Vivace ICP reaches 7.77\,s and 2.63\,mUSD at 12 months, still $2.34\times$ faster than Athena SCD2 at 53\% of the cost. Since cold starts are the common condition under intermittent access, this advantage is of practical significance.

\paragraph{Always-on baselines}
We compare against two always-on baselines, each with a distinct advantage: ClickHouse~\cite{clickhouse} for fast columnar scans, and MobilityDB~\cite{mobilitydb} for native interval semantics that simplify query writing. Both hold the 12-month SpotLake history on an r7i.xlarge (4\,vCPU, 32\,GiB) instance, the specification recommended by the official ClickHouse documentation.

ClickHouse's infrastructure cost is fixed at about 201\,USD per month regardless of the query count. The break-even frequency at which its per-query cost equals Vivace ICP is about 250 queries per hour at the 12-month window and about 1,700 per hour at the 1-month window, and below those rates serverless is cheaper.
MobilityDB's native temporal types and aggregates such as \texttt{twAvg} simplify query writing, but execution is slower, with a 12-month 10-query average of 7.86\,s, about 2.6$\times$ Vivace ICP (2.98\,s). Against either advantage, Vivace delivers practical latency while keeping a pay-per-query cost that scales only with execution frequency.
 
\subsection{Operator Generality (RQ5)}
\label{sec:eval-generality}
For RQ5, we evaluate operator families on an independent SCD Type 2 benchmark with an entirely different schema.
 
\paragraph{MobilityDB benchmark}
MobilityDB-TPCDS~\cite{mobilitydb-tpcds} is a public temporal warehouse workload that provides the historical dimension tables of the TPC-DS~\cite{tpcds} benchmark in three implementations: SCD Type 2, temporal data warehouse, and MobilityDB temporal types. This evaluation uses its \texttt{scd\_item} table (SF100). Each row is an SCD Type 2 record holding the $[\From,\To)$ interval during which attributes of one item, such as brand, class, and price, persisted, a schema entirely different from SpotLake.
 
The six temporal algebra queries of this benchmark (benchmark notation Q1--Q6, hereafter MQ1--MQ6) are all expressed as compositions of Vivace operator families, requiring no new operator (Table~\ref{tab:eval-mobilitydb}). Correctness was judged by an independent oracle built from the original \texttt{scd\_item} without Vivace primitives. For all six queries, the fan-out results matched this oracle and the single-process reference in canonical form, and MQ1--MQ4 and MQ6 also matched native MobilityDB.
 
MQ5 computes a temporal difference followed by a key-domain restriction. It finds the parts of a brand's intervals where the price is at or below a threshold, restricted to items whose price ever exceeded it. With brand intervals $B$ and threshold-exceeding intervals $P_i$ per item, the core computation is $B \setminus \bigcup_i P_i$ plus an item-level existence condition. Implementing this difference directly with interval set operations is error-prone at two steps: collecting and subtracting the multiple $P_i$, and clipping the result back inside $B$. The benchmark's three official SQL implementations indeed diverge at exactly these points. \texttt{Q5\_MobDB} subtracts $\bigcap_i P_i$ instead of $\bigcup_i P_i$, retaining intervals that should be removed whenever two or more $P_i$ exist, and \texttt{Q5\_TDW} and \texttt{Q5\_SCD} do not clip the result to $B$, so their results extend outside it. These errors affect 19 of the 94 target items at SF100.

In Vivace, owing to the full-state same-base intervals, this difference reduces to a single boolean over one interval. Evaluating ``the brand condition holds and price $\le$ threshold'' on each interval suffices, and since the intervals are already clipped into closed segments, neither error point arises. The item-level existence condition is applied once by the key-domain semi-join (\S\ref{sec:merge-laws} L7) after both windows are complete. As a result, against the same oracle, only Vivace was correct on all 94 items (94/94), while \texttt{Q5\_MobDB} reached 93 and \texttt{Q5\_TDW} and \texttt{Q5\_SCD} reached 75.
 
These results validate that Vivace's operator families express temporal algebra over SCD Type 2 schemas in general, not over a specific dataset, without new primitives.
 
\begin{table}[t]
\centering
\caption{MobilityDB-TPCDS MQ1--MQ6 as compositions of Vivace operators, each discharged by a merge law.}
\label{tab:eval-mobilitydb}
\setlength{\tabcolsep}{3pt}
\begin{tabularx}{\columnwidth}{@{}l >{\raggedright\arraybackslash}X l l@{}}
\toprule
\textbf{Query} & \textbf{Vivace operator composition} & \textbf{Carrier} & \textbf{Law} \\
\midrule
MQ1 & predicate window (SO1) $\to$ coalesce & $W$ & L1 \\
MQ2 & projection $\to$ coalesce & $W$ & L1 \\
MQ3 & same-base boolean and (SO3-I) $\to$ coalesce & $W$ & L1 \\
MQ4 & same-base boolean or (SO3-U) $\to$ coalesce & $W$ & L1 \\
MQ5 & difference (SO3-D) $+$ key-domain finalizer & $(W, W_\psi)$ & L7 \\
MQ6 & count timeline (SO4) $\to$ threshold & $C_{step}$ & L3,\,L4 \\
\bottomrule
\end{tabularx}
\end{table}

\section{Related Work}
\paragraph{Temporal databases and analytics}
Sequenced semantics defines the meaning of a temporal query through the per-time-point interpretation of the relation, and coalescing arranges interval results with differing boundaries into a comparable canonical form~\cite{sql2011,jensen-unification,dignos,dignos-sequenced,dignos-multiset}. Access methods for time-evolving data are surveyed in~\cite{temporal-access-survey}, and parallel temporal aggregation has been studied on shared-nothing architectures~\cite{parallel-agg-icde99}. Timeline Index~\cite{timeline-index}, ParTime~\cite{partime}, and temporal ranking~\cite{temporal-ranking} accelerate temporal aggregation, joins, and aggregate top-$k$ over query intervals efficiently, but all assume an in-memory engine. MobilityDB~\cite{mobilitydb} extends PostgreSQL with native temporal types and operators, but requires an always-on database server.

\paragraph{Interval histories and versioned tables}
SCD Type 2~\cite{scd} stores entity versions with their start and end times, and TPC-DS~\cite{tpcds} and MobilityDB-TPCDS~\cite{mobilitydb-tpcds} provide historical dimensions as benchmarks. Lakehouse features such as Iceberg time travel and Delta CDF likewise yield valid-from/to rows from table versions and row-level changes~\cite{iceberg-time-travel,delta-cdf,snowflake-change,dex}. However, exact query-window analytics over these histories remains unaddressed.
 
\paragraph{Object-store SQL and serverless analytics}
Athena~\cite{athena} and BigLake~\cite{biglake} query object storage with SQL, but users must hand-write the interval overlap and clipping logic. Columnar engines such as ClickHouse~\cite{clickhouse} also scan object storage directly but require an always-on cluster. Lambada~\cite{serverless-db-lambada}, Starling~\cite{scalable-query-starling}, and Cackle~\cite{cackle} provide serverless fan-out execution but target general relational queries, lacking the temporal handling that interval histories require. 

\paragraph{Mergeable and incremental computation}
MapReduce combiners, Spark RDDs, the Dataflow model, and incremental view maintenance~\cite{mapreduce,spark-rdd,dataflow,mergeable-summaries,dbtoaster,dbsp,fivm} established the principle of combining local summaries or deltas later. Vivace applies this principle to temporal operators over intervals: associative aggregates such as duration and count merge their partials directly, while non-associative operations such as ranking record the comparator and shard placement in the plan to reproduce the exact result.
 
Table~\ref{tab:related-comparison} contrasts Vivace with prior systems.

\begin{table}[t]
  \centering
  \setlength{\tabcolsep}{2pt}
  \renewcommand{\arraystretch}{1.25}
  \scriptsize
  \caption{Comparison with prior systems. Always-on temporal covers Timeline Index, ParTime, and MobilityDB.}
  \label{tab:related-comparison}
  \begin{tabular*}{\columnwidth}{@{\extracolsep{\fill}}lccccc@{}}
    \toprule
    & \textbf{\shortstack{Always-on\\temporal}} & \textbf{ClickHouse} & \textbf{\shortstack{Athena/\\BigLake}} & \textbf{\shortstack{Lambada/\\Starling}} & \textbf{Ours} \\
    \midrule
    \makecell[l]{\textbf{Temporal operators on}\\ \textbf{interval histories}}
      & \checkmark & -- & -- & -- & \checkmark \\
    \makecell[l]{\textbf{Queries data on} \textbf{object storage}}
      & -- & \checkmark & \checkmark & \checkmark & \checkmark \\
    \textbf{Pay-per-query execution}
      & -- & -- & \checkmark & \checkmark & \checkmark \\
    \makecell[l]{\textbf{Exact partitioned execution}\\ \textbf{w/o state exchange (R1, R2)}}
      & -- & -- & -- & -- & \checkmark \\
    \bottomrule
  \end{tabular*}
\end{table}

\section{Conclusion and Future Work}
 
This paper presented Vivace, a system that executes temporal OLAP over interval histories across independent serverless functions.
By cutting the source relation into a time-partitioned layout before queries, each function obtains its complete input from a single file, and the reducer merges the per-function Carriers by operator-specific rules. We proved as a theorem that under these two requirements the result is exactly that of single-process execution.
Against a SQL baseline on three real datasets with different interval characteristics, this design removed the input preparation cost that grows with history depth, reducing latency by up to 82\% and cost by up to 84\%.
 
Vivace currently covers queries over a single interval relation or its join with time-invariant dimensional attributes. Temporal joins across relations, dynamically constructed cohorts, and sequential or sliding-window patterns lie outside this scope, and the planner rejects such requests.
Future work includes defining a multi-relation co-partitioned closure to support restricted temporal joins, introducing boundary-state carriers to cover some sliding-window patterns, and selecting the chunk size and layout automatically for a given workload.

\section*{Acknowledgment}
We used Anthropic's Claude Code and OpenAI's Codex to assist with the system implementation of this work. All AI-generated code was reviewed and verified by the authors. 

\bibliographystyle{IEEEtran}
\bibliography{vivace}

\end{document}